\documentclass[acmsmall,screen]{acmart}

\usepackage{amsmath}
\usepackage{stmaryrd}
\usepackage{pifont}
\usepackage{booktabs}
\usepackage{tabularx}
\usepackage{array}
\usepackage{multirow}
\usepackage{enumitem}
\usepackage{algorithm}
\usepackage{algpseudocode}
\usepackage{listings}
\usepackage{xspace}
\usepackage{tikz}
\usepackage{fontawesome5}
\usepackage{pgfplots}
\pgfplotsset{compat=1.18}
\usepgfplotslibrary{groupplots}
\definecolor{tpFalse}{HTML}{B4472A}
\definecolor{tpAttr}{HTML}{1F5C8B}
\definecolor{tpHuman}{HTML}{6E7B8B}
\definecolor{tpLight}{HTML}{C9D6E0}
\definecolor{tpDS}{HTML}{2C6E8F}
\definecolor{tpQW}{HTML}{C08A2E}
\definecolor{tpOk}{HTML}{4A7C59}
\pgfplotsset{
  tpBase/.style={
    font=\footnotesize,
    tick label style={font=\scriptsize},
    label style={font=\scriptsize},
    legend style={font=\scriptsize, draw=none, fill=none},
    axis line style={draw=black!55},
    tick style={draw=black!55},
    grid style={draw=black!12},
  },
}
\pgfplotsset{
  tpFig/.style={
    tpBase,
    title style={font=\footnotesize\bfseries, yshift=1.5pt},
    ymajorgrids,
    axis y line*=left,
    tick align=outside,
    major tick length=2pt,
    legend style={
      at={(0.5,-0.185)}, anchor=north, legend columns=-1,
      draw=none, fill=none, font=\scriptsize,
      /tikz/every even column/.append style={column sep=5pt},
    },
    legend image post style={scale=0.68},
  },
}
\usepackage{float}
\usepackage{needspace}
\usetikzlibrary{arrows.meta,positioning,fit,backgrounds,calc}
\usetikzlibrary{patterns,patterns.meta}
\usepackage{placeins}

\lstdefinelanguage{ACL2}{
  morekeywords={defun,defthm,defthmd,defxdoc,if,endp,car,cdr,cons,nil,t,include-book,implies,and,or,equal,let,declare,xargs,skip-proofs},
  sensitive=true,
  morecomment=[l]{;},
  morestring=[b]"
}
\lstdefinelanguage{Lean}{
  morekeywords={def,theorem,by,rw,simp,rfl,String,List,Nat,Bool,fun,match,with},
  sensitive=true,
  morecomment=[l]{--},
  morestring=[b]"
}
\newcommand{\tool}{DocAuth\xspace}

\newcolumntype{Y}{>{\raggedright\arraybackslash}X}
\newcolumntype{R}{>{\raggedleft\arraybackslash}X}
\newcolumntype{L}{>{\hsize=1.8\hsize}>{\raggedright\arraybackslash}X}
\newcolumntype{N}{>{\hsize=0.8\hsize}>{\raggedleft\arraybackslash}X}
\newcommand{\tpTableSetup}{\footnotesize\setlength{\tabcolsep}{5pt}\renewcommand{\arraystretch}{1.22}}

\newtheorem{definition}{Definition}
\newtheorem{lemma}{Lemma}

\newcommand{\Valid}{\mathsf{Valid}}

\title[Beyond Correctness: Attribution of Natural-Language Documentation]{Beyond Correctness: Attribution of Natural-Language Documentation for Formal Libraries}

\renewcommand\footnotetextcopyrightpermission[1]{}
\setcopyright{none}
\acmDOI{}
\acmISBN{}
\author{Sirui Jiang}
\affiliation{%
  \institution{Zhejiang Sci-Tech University}
  \city{Hangzhou}
  \state{Zhejiang}
  \country{China}}
\email{2024337621116@mails.zstu.edu.cn}
\author{Dongming Xiang}
\affiliation{%
  \institution{Zhejiang Sci-Tech University}
  \city{Hangzhou}
  \state{Zhejiang}
  \country{China}}
\email{dmxiang@zstu.edu.cn}
\author{Shichao Wu}
\affiliation{%
  \institution{Zhejiang Sci-Tech University}
  \city{Hangzhou}
  \state{Zhejiang}
  \country{China}}
\email{2024210602019@mails.zstu.edu.cn}
\author{Yuxuan Zhou}
\affiliation{%
  \institution{Zhejiang Sci-Tech University}
  \city{Hangzhou}
  \state{Zhejiang}
  \country{China}}
\email{2024316101003@mails.zstu.edu.cn}
\author{Yihong Wu}
\affiliation{%
  \institution{Zhejiang Sci-Tech University}
  \city{Hangzhou}
  \state{Zhejiang}
  \country{China}}
\email{yihongwu88@126.com}
\author{Fangze Zhu}
\affiliation{%
  \institution{Zhejiang Gongshang University}
  \city{Hangzhou}
  \state{Zhejiang}
  \country{China}}
\email{Zhufangz2402090131@163.com}
\author{Wang Lin}
\affiliation{%
  \institution{Zhejiang Sci-Tech University}
  \city{Hangzhou}
  \state{Zhejiang}
  \country{China}}
\email{linwang@zstu.edu.cn}
\author{Guanjun Liu}
\affiliation{%
  \institution{Tongji University}
  \city{Shanghai}
  \country{China}}
\email{liuguanjun@tongji.edu.cn}

\begin{document}

\begin{abstract}
Formal languages are widely used in software engineering to specify and verify program properties. Because their mathematical notation and logic are hard to read, developers often rely on natural-language documentation to understand the functions in formal libraries. Yet a documentation statement can be provably true and still hold after the function's behavior changes.
We call a statement's dependence on the function's specific behavior attribution, and we give a method to check it.
We design and implement DocAuth, a documentation-generation method that uses attribution to decide what to publish and attaches to each published sentence the formal property recovered from it, together with checkable evidence.
On $150$ functions from the Community Books library of the ACL2 theorem prover, three large language models produce proved but non-attributable properties for about a third of the functions. DocAuth matches or exceeds the attribution rate of a mutation-based filter while delivering $37\%$ more attributable sentences. Compared with a pipeline that checks truth alone, it raises the share of published sentences that a second model confirms as attributable from $82.9\%$ to $94.7\%$.
These results show that provable truth does not guarantee that documentation reflects what a function actually does. Attribution complements truth as a criterion for generating and evaluating documentation for formal libraries.
\end{abstract}

\maketitle
\pagestyle{plain}
\thispagestyle{plain}

\section{Introduction}
Establishing that a program behaves as intended is a fundamental task in software engineering, especially for systems with stringent reliability requirements. Formal languages let developers state the properties a program should satisfy in precise mathematical notation and establish them by mechanized proof, and they have been applied to processors and systems software~\citep{hunt2017industrial}. Theorem provers such as ACL2~\citep{kaufmann2000computer}, Lean~\citep{demoura2015lean}, Coq~\citep{coqmanual}, and Isabelle/HOL~\citep{nipkow2002isabelle} have accumulated large libraries of formally defined functions and proved theorems. The notation, however, is hard to read: to work out what a function does, a reader often has to unfold nested definitions, trace recursion, and piece together behavior that several theorems characterize only jointly. Formal libraries therefore need natural-language documentation alongside their definitions and proofs, so that developers can understand and reuse their functions~\citep{davis2014industrial}. Now that large language models (LLMs) generate documentation and properties at scale, the generated content is usually checked for correctness: whether it can be proved, or whether it is consistent with the code~\citep{ma2025specgen,sun2023clover}.

\begin{figure}[t]
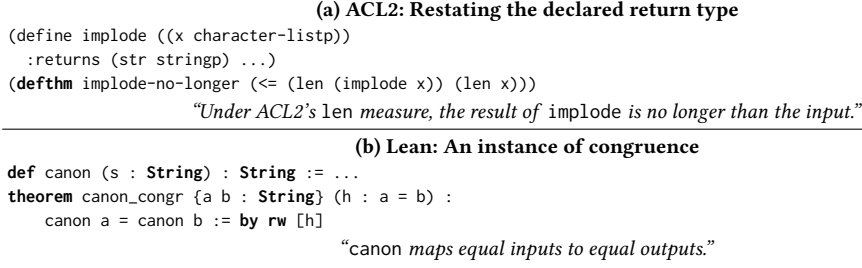

\footnotesize
\lstset{basicstyle=\ttfamily\scriptsize,frame=none,aboveskip=1pt,belowskip=1pt,
        xleftmargin=2pt}
\newcommand{\exsep}{\par\vspace{2pt}\hrule\vspace{3pt}}
\newcommand{\exsay}[1]{\emph{“#1”}\par}

\textbf{(a) ACL2: Restating the declared return type}\par
\begin{lstlisting}[language=ACL2]
(define implode ((x character-listp))
  :returns (str stringp) ...)
(defthm implode-no-longer (<= (len (implode x)) (len x)))
\end{lstlisting}
\exsay{Under ACL2's \texttt{len} measure, the result of \texttt{implode} is no longer than the input.}
\exsep

\textbf{(b) Lean: An instance of congruence}\par
\begin{lstlisting}[language=Lean]
def canon (s : String) : String := ...
theorem canon_congr {a b : String} (h : a = b) :
    canon a = canon b := by rw [h]
\end{lstlisting}
\exsay{\texttt{canon} maps equal inputs to equal outputs.}

\caption{Proved properties that hold regardless of the target's implementation}
\Description{Two code snippets and their corresponding documentation statements. The upper part shows the definition of \texttt{implode} in ACL2 and a theorem stating that the length does not increase; the lower part shows the definition of \texttt{canon} in Lean and a theorem stating that equal inputs produce equal outputs.}
\label{fig:motivating}
\end{figure}

A documentation statement can, however, be provably true and still say nothing about the target function's specific behavior. Figure~\ref{fig:motivating} shows two examples. For \texttt{implode} from the ACL2 Community Books, the statement ``Under ACL2's \texttt{len} measure, the result of \texttt{implode} is no longer than the input'' is backed by a proved theorem. But \texttt{implode} returns a string, and \texttt{len} returns $0$ for strings, so every implementation that returns a string satisfies the statement. The statement that the Lean function \texttt{canon} maps equal inputs to equal outputs follows from equality substitution alone and holds for every function of the same type. Both statements paraphrase their theorems accurately, yet one conveys only a return type and the other only a general fact.

We call a statement's dependence on the target function's specific behavior \emph{attribution}. To check it, we hold the function's public interface (its input constraints and return type) and the background knowledge fixed and vary only the target's behavior. A statement that holds for the target is \emph{attributable} if it fails for at least one alternative implementation that satisfies the same interface; such an implementation, together with an input on which the statement fails, is a \emph{witness}. A statement that follows from the interface and background knowledge alone is \emph{non-attributable}: even when formally proved, it tells the reader nothing about the target's behavior beyond the interface.

Attribution thus complements truth. Truth confirms that a statement holds for the target; attribution confirms that it describes behavior beyond the interface. Since callers can already see the interface, a statement about the target's behavior should say more than the interface does, which makes attribution a basic requirement for such statements. It is one quality dimension, orthogonal and complementary to importance, completeness, and readability, which together determine whether documentation is useful to developers; we focus on attribution.

Attribution answers a question that no existing measure answers: does a statement tell the reader anything about behavior beyond the interface (Table~\ref{tab:conceptual-comparison})? It does so by measuring each statement against the function's public interface. Every alternative implementation must satisfy this interface, so the alternatives range over precisely the behaviors a caller cannot rule out without reading the function body. Existing methods also test a conclusion by changing something, but none of them bounds the change by the interface. Mutation-based evaluation of specifications~\citep{endres2024nl2postcond,amaral2025mutdafny} changes too little: it tries only a handful of edits to the function body, and these edits often leave the behavior a property describes untouched and need not satisfy the interface. Vacuity checking and semantic dependence~\citep{beer2001vacuity,podgurski1990dependence} change too much: they replace a part with an unconstrained value and thus break the declared interface as well, so a statement that merely restates the return type would be ``falsified'' even though it adds nothing beyond the interface. Anchoring the check to the interface also backs both verdicts with checkable evidence. Non-attribution is established by a proof that covers every implementation satisfying the interface, which no finite set of mutants can provide; attribution is confirmed by a witness that itself satisfies the interface, which readers can check for themselves.

\begin{table}[!htbp]
\centering
\tpTableSetup
\caption{Attribution compared with existing measures. \ding{51} guaranteed; \faAdjust depends on the selected variants; \ding{55} not guaranteed.}
\label{tab:conceptual-comparison}
\setlength{\tabcolsep}{4pt}
\begin{tabularx}{\linewidth}{@{}>{\raggedright\arraybackslash}p{0.25\linewidth}>{\raggedright\arraybackslash}p{0.13\linewidth}>{\raggedright\arraybackslash}p{0.18\linewidth}Y>{\centering\arraybackslash}p{0.10\linewidth}@{}}
\toprule
\textbf{Measure} & \textbf{Unit of decision} & \textbf{Reference} & \textbf{Question answered} & \textbf{Can reject Figure~\ref{fig:motivating}\newline statements} \\
\midrule
Provability~\citep{ma2025specgen} & Individual property & Original target & Can the property be proved? & \ding{55} \\
Mutation discrimination~\citep{endres2024nl2postcond,amaral2025mutdafny,goldweber2024ironspec} & Property or whole specification & Selected code mutants & Which tested implementations does it rule out? & \faAdjust \\
Vacuity and semantic dependence~\citep{beer2001vacuity,podgurski1990dependence,chockler2008causes} & Property or program part & Changes to a subformula or program part & Does the conclusion depend on the changed part? & \ding{55} \\
Document--code equivalence~\citep{eghbali2026documentary,zeng2025veriequivbench} & Entire document or specification & Overall behavior of the original target & Do the two describe the same behavior? & \ding{55} \\
Source support and consistency~\citep{rashkin2023measuring,panthaplackel2021deep,sun2023clover} & Sentence or comment & Source text or code & Is the statement supported and consistent? & \ding{55} \\
\midrule
\textbf{Attribution (ours)} & \textbf{Individual statement} & \textbf{Implementations under the fixed interface} & \textbf{Does it add behavioral information beyond the interface?} & \ding{51} \\
\bottomrule
\end{tabularx}
\end{table}

Turning attribution into a check and using it for documentation poses three challenges. First, non-attribution is a claim about every implementation that satisfies the interface, including implementations not definable in the logic. We reduce it to a single proof in an \emph{erased environment}, where the target is replaced by a fresh symbol constrained only by its interface (Lemma~\ref{lem:erasure}). Second, a witness must itself satisfy the interface, so arbitrary edits do not qualify. DocAuth builds alternative implementations from fixed templates over the function's arguments and has ACL2 prove that each one satisfies the interface. Third, documentation is prose, and rendering a property as a sentence can drop precisely the part that depended on the implementation. DocAuth therefore recovers a formal property from each rendered sentence, checks it again, and publishes only the sentences that pass, each with its recovered property and evidence.

On $150$ functions from the ACL2 Community Books, three LLMs produce proved but non-attributable properties for about a third of the functions. Attributable properties also catch changes their check never saw: under interface-preserving code changes that differ from the witness templates, $64.8\%$ of attributable properties detect a change, compared with $15.3\%$ of the proved properties for which no witness is found and none of the non-attributable ones. At the sentence level, DocAuth raises the share of published sentences confirmed as attributable from $82.9\%$, for a pipeline that checks truth alone, to $94.7\%$, and stays $9.3$ points ahead when the reassessment model cannot see the code. It also delivers $37\%$ more attributable sentences than a mutation-based filter. The two checks divide the work: rechecking the rendered sentences accounts for most of the rate gain, and selecting attributable candidates raises the number of attributable sentences.

\Needspace{4\baselineskip}
This paper makes three contributions:
\begin{itemize}[leftmargin=*]
 \item \textbf{A definition of attribution and a sound check for it.} Relative to the public interface, we separate what the interface already states from what an individual statement adds. An erasure proof establishes non-attribution for all conforming implementations, and an interface-conforming witness confirms attribution.
 \item \textbf{DocAuth, a documentation-generation method.} It applies the check to candidate properties and to the properties recovered from rendered sentences, and publishes each sentence with its formal property and evidence.
 \item \textbf{An empirical evaluation.} Non-attributable properties occur in about a third of the functions, and attributable properties detect interface-preserving changes that the others miss. Pooled over the three generator models, DocAuth delivers the most attributable sentences at the highest attribution rate among the methods compared.
\end{itemize}

\section{Attribution}
\label{sec:attribution}

Consider \texttt{str::base64-encode-list} from the ACL2 Community Books, one of our evaluation targets, which encodes a character list in Base64. Among the candidate properties an LLM proposed for it are these two:

\begin{lstlisting}[language=ACL2]
(character-listp (base64-encode-list a1))
(equal (base64-encode-list nil) nil)
\end{lstlisting}

ACL2 proves both, but only the second says something about this encoder. The first repeats the \texttt{:returns} declaration and holds for any implementation that returns a character list. Replacing the implementation exposes the difference: an implementation that always returns a one-character list still satisfies the interface, and the first property remains true while the second fails.

\subsection{Formal Definition}

\begin{definition}[Attribution, non-attribution, and witnesses within a search scope]
\label{def:attribution}
Let $I$ be the public interface of the target $f$, and let the background environment $B$ contain, in place of $f$, only a placeholder symbol constrained by $I$, with no fact about that symbol other than $I$. Let $\mathcal{C}(I)$ be the set of all functions that can be substituted for the placeholder and satisfy $I$, including functions not definable in ACL2; $f$ satisfies $I$, so $f\in\mathcal{C}(I)$. The property $\varphi$ is relativized to the input constraints of $I$; $\mathcal{V}_{\varphi}$ is the set of assignments to its free variables over the domain of discourse; and $\llbracket\varphi[\nu]\rrbracket_{g}\in\{\top,\bot\}$ denotes the truth value of $\varphi$ in $B$ when $g$ is substituted for the placeholder and $\nu$ for the free variables. If $\llbracket\varphi[\nu]\rrbracket_{f}=\top$ for every $\nu\in\mathcal{V}_{\varphi}$, we write $\Valid(f,\varphi)$, and we let
\[
 W(\varphi)=\bigl\{(g,\nu)\in\mathcal{C}(I)\times\mathcal{V}_{\varphi}
   \bigm|\llbracket\varphi[\nu]\rrbracket_{g}=\bot\bigr\}.
\]
Under $\Valid(f,\varphi)$, we call $\varphi$ \emph{attributable} if $W(\varphi)\neq\varnothing$ and \emph{non-attributable} if $W(\varphi)=\varnothing$. We call each $(g,\nu)\in W(\varphi)$ a \emph{witness}; by $\Valid(f,\varphi)$, every witness has $g\neq f$.
Within a finite search scope $\mathcal{S}\subseteq\mathcal{C}(I)\times\mathcal{V}_{\varphi}$, we \emph{confirm attribution} if $\mathcal{S}\cap W(\varphi)\neq\varnothing$; if $\mathcal{S}\cap W(\varphi)=\varnothing$ and $\varphi$ has not been established as non-attributable, we say that \emph{no witness is found} for $\varphi$ within that scope.
\end{definition}

A target proof and a single witness suffice to confirm attribution, but a finite search that finds no witness cannot establish non-attribution. Our checker uses $\mathcal{S}=A(I)\times X_{\varphi}$, where $A(I)\subseteq\mathcal{C}(I)$ and $X_{\varphi}\subseteq\mathcal{V}_{\varphi}$ are finite sets generated by fixed rules.

\subsection{Checking Method}
\label{sec:audit-outcomes}
The check is anchored to the interface $I$ published in the source code: proofs about the original function and the witness search over alternative implementations both assume the input constraints of $I$ (the \texttt{guard}).

The erased environment $E_f$ implements the background environment $B$ of Definition~\ref{def:attribution}. Rather than deleting the consequences of $f$ from the original environment, we build $E_f$ by loading three parts in order: the events that precede the definition of $f$ in its source file, a set of standard background libraries, and finally a fresh symbol $u_f$. The only constraint on $u_f$ is the return contract in $I$ (the declared return type); the input constraints become premises of the property. Because $u_f$ is introduced after every other event, no fact in $E_f$ other than the return contract mentions it, and no lemma needs to be filtered out. During checking, occurrences of $f$ in a property are rewritten to $u_f$. If the property also uses other functions that call $f$, their meaning still depends on $f$, so the property is left undecided in $E_f$ and recorded as unresolved.

\begin{lemma}[Soundness of erasure verdicts]
\label{lem:erasure}
Suppose that no function in $\varphi$ other than $f$ depends on $f$. If $\varphi[u_f/f]$ is provable in $E_f$, then $W(\varphi)=\varnothing$.
\end{lemma}

\noindent\emph{Proof sketch.} Every event in $E_f$ is a definition or has been proved by ACL2, so it holds in the standard model. Take any $g\in\mathcal{C}(I)$: the only fact in $E_f$ that mentions $u_f$ is the return contract, which $g$ satisfies, so interpreting $u_f$ as $g$ keeps every fact of $E_f$ true. By soundness, $\varphi[u_f/f]$ is true under this interpretation; by the premise, $\varphi$ involves $f$ only at the rewritten positions, so $\llbracket\varphi[\nu]\rrbracket_{g}=\top$ for every $\nu\in\mathcal{V}_{\varphi}$. When $g$ is definable in ACL2, this argument coincides with ACL2's functional instantiation rule~\citep{kaufmann2001structured}. $\square$

Checking starts in $E_f$; properties not proved there move on to the original environment $T$, which consists of the same background libraries plus the library containing $f$. A property refuted by a counterexample in $T$ is recorded as false for the target and receives no attribution verdict. Every other property whose check completes receives exactly one of the following outcomes.

\begin{itemize}[leftmargin=*]
\item \textbf{\ding{182}\;Attributable.} $\varphi$ is proved in $T$, and $A(I)\times X_{\varphi}$ contains an alternative implementation $g$ satisfying the interface and an assignment $\nu$ that together falsify $\varphi$. The target proof and the witness form the \emph{attribution evidence}.
\item \textbf{\ding{183}\;Non-attributable.} $\varphi$ remains provable in $E_f$; that is, $I$ and the background already entail $\varphi$. By Lemma~\ref{lem:erasure}, $W(\varphi)=\varnothing$, and since $f\in\mathcal{C}(I)$, $\Valid(f,\varphi)$ follows without a proof in $T$.
\item \textbf{\ding{184}\;No witness found.} $\varphi$ is not proved in $E_f$ but is proved in $T$, and $A(I)\times X_{\varphi}$, searched in full, contains no witness.
\end{itemize}

If a dependency is unavailable, a proof or evaluation does not complete, or no sampled input satisfies the property's premises, the outcome is \emph{unresolved}. Figure~\ref{fig:erasure} shows how the two properties above receive outcomes \ding{183} and \ding{182}, respectively.

\begin{figure}[!htbp]
\footnotesize
\lstset{basicstyle=\ttfamily\scriptsize,frame=none,aboveskip=1pt,belowskip=1pt,
        xleftmargin=2pt}
\hrule\vspace{3pt}
{\raggedright
\textbf{Erased environment.} A fresh constrained symbol replaces the target
\texttt{str::base64-encode-list}; the symbol carries only the target's public interface
\texttt{:returns (x-enc character-listp)}:
\par}
\begin{lstlisting}[language=ACL2]
(encapsulate (((u-56c775 *) => *))
  (local (defun u-56c775 (a1) (declare (ignore a1)) nil))
  (defthm u-56c775-contract
    (character-listp (u-56c775 a1)) :rule-classes :rewrite))
\end{lstlisting}
\vspace{2pt}
{\raggedright
\textbf{Outcomes.} Both properties are rewritten over this symbol, with the input constraint as a premise. The first is still provable, because the
\texttt{:returns} declaration entails it (\textbf{non-attributable}); the second is not, and a witness that satisfies the interface falsifies it (\textbf{attributable}).
\par}
\begin{lstlisting}[language=ACL2]
(implies (character-listp a1)
         (character-listp (u-56c775 a1)))   ; after erasure: proved
(implies (character-listp nil)
         (equal (u-56c775 nil) nil))        ; after erasure: not proved
(defun g1-21f79 (a1) '(#\a))                ; witness: satisfies the interface and falsifies the property
\end{lstlisting}
\vspace{3pt}\hrule
\caption{The erased environment and the attribution check for the running example}
\Description{An ACL2 code example in two parts. The first part replaces the target encoding function with a constrained symbol carrying only a character-list return contract. The second part rewrites two candidate properties over this symbol: the property restating the return contract remains provable after erasure, showing that the interface already entails it, and is classified as non-attributable; the property fixing the result on empty input is falsified by a constant implementation satisfying the same interface and is classified as attributable.}
\label{fig:erasure}
\end{figure}

\section{DocAuth}
\label{sec:docauth}

Attribution is decided on formal properties, but developers read natural-language sentences. DocAuth therefore selects documentation content at the formal level, renders it as sentences, and decides attribution again for what each sentence actually says (Figure~\ref{fig:docauth-pipeline}). The workflow has three steps. An LLM proposes candidate properties $\varphi$ for the target function, and the candidate check keeps only the attributable ones. Each kept property is rendered as a sentence $s$, from which DocAuth recovers a formal property $\psi$; $\psi$ is checked again, and $s$ is published only if $\psi$ is attributable. Both checks run Algorithm~\ref{alg:docauth}, with $\varphi$ and $\psi$, respectively, as the property of Definition~\ref{def:attribution}.

The two checks play different roles. The candidate check decides what to write, so that the limited sentence budget is spent on attributable content. The sentence recheck decides what to publish, because readers read sentences, not properties. Both checks let through only content confirmed as attributable; nothing is published when the outcome is unresolved or no witness is found. Every published sentence comes with the property recovered from it and that property's attribution evidence, which readers can check for themselves.

\begin{figure}[!htb]
  \centering
  \includegraphics[width=0.98\textwidth]{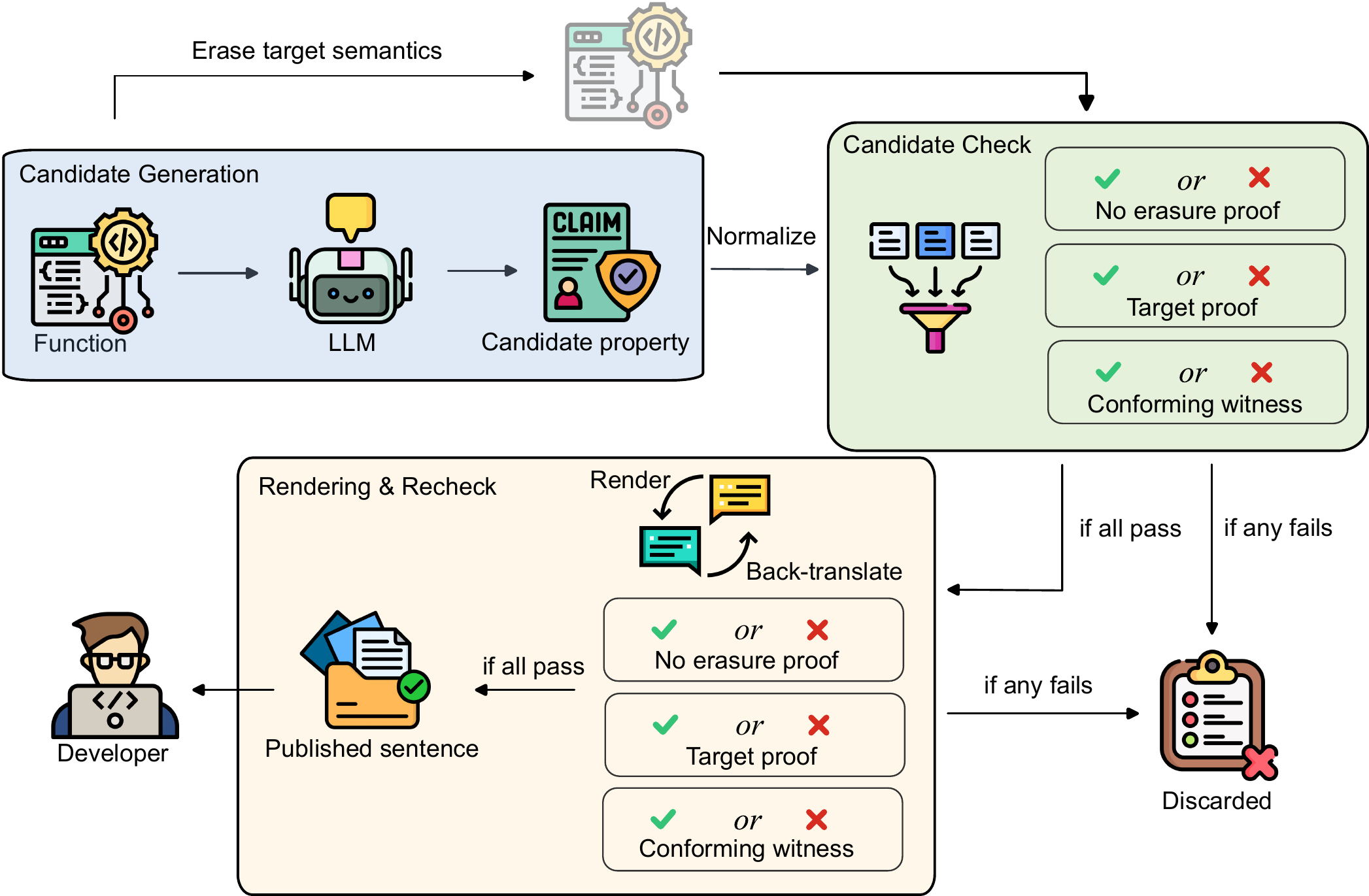}
  \caption{The DocAuth workflow}
  \Description{The workflow has three stages. During candidate generation, an LLM reads the formal target function and proposes candidate properties. The candidate check first identifies non-attributable properties through an erasure proof, then obtains a target proof and an interface-conforming witness for the remaining properties; these two items form the attribution evidence. Properties that pass the check are rendered as natural-language sentences, from which formal properties are recovered by back-translation and subjected to the sentence recheck; only sentences that pass both checks are published to developers, and everything else is discarded.}
  \label{fig:docauth-pipeline}
\end{figure}

\subsection{Candidate Generation}

Candidate generation is an ordinary property-generation step. The model sees the formal source, including the complete function body, so that it can summarize implemented behavior worth documenting. The prompt asks for the shape of the result, relationships between the result and the inputs, boundary cases, and behavior under composition with standard functions; it mentions neither attribution nor erasure (Figure~\ref{fig:prompt}). All content filtering is left to the later checks, so the generator can be swapped for another model or property-generation method. Identifiers defined in the corpus are replaced with opaque names so that the model describes behavior from the definitions rather than from names; ACL2 built-in vocabulary is kept. The model is queried once per function for as many candidates as the candidate budget specifies. Responses are normalized to single ACL2 terms (stripping wrappers such as \texttt{defthm}, removing duplicates, and restoring the original names) and stored in the order the model returns them.

\begin{figure}[!htb]
\footnotesize
\hrule\vspace{3pt}
\textbf{Propose (formal source $\to$ candidate property $\varphi$).}
You are preparing formal evidence for natural-language documentation of a function in an ACL2 library so that other users can understand its behavior without reading its definition.
Output only a numbered list with one ACL2 term per line; do not include explanatory text, \texttt{defthm} wrappers, or code fences.
Each line must be a well-formed ACL2 term about the function shown, using only that function and standard ACL2 primitives.
Cover the shape of the result, relationships between the result and inputs, boundary cases, and behavior when composed with standard functions.
\vspace{3pt}\hrule
\caption{Prompt for candidate-property generation}
\Description{The candidate-generation system prompt in plain text. It asks the model to propose well-formed ACL2 property terms for the function shown, covering the shape of the result,
relationships between the result and inputs, boundary cases, and behavior when composed with standard functions.}
\label{fig:prompt}
\end{figure}

\subsection{Implementing the Check}
\label{sec:audit}

Before generating candidates, DocAuth extracts the target's explicit \texttt{guard} and any separately declared return contract from the source as the interface $I$ and freezes it. The reference is thus the interface that the library's authors have already published, not one chosen by DocAuth or the model. The erased environment $E_f$ and the original environment $T$ are built as in Section~\ref{sec:audit-outcomes}, and all checks for a given target share the same interface and environments. Checked properties are not registered as proof rules, so checks of different properties do not affect one another. Algorithm~\ref{alg:docauth} gives the complete procedure, where $[g/f]$ replaces the target $f$ with $g$ and $[\nu]$ substitutes the concrete inputs $\nu$.

\begin{algorithm}[!htb]
\caption{Determining whether a property is attributable}
\label{alg:docauth}
\footnotesize
\begin{algorithmic}[1]
\Require Original environment $T$, target $f$, property $\varphi$, fixed alternative templates and input-generation rules
\Ensure Attribution evidence (target proof and witness), or a reason to withhold publication
\State $I\gets$ public interface of $f$
\State $E_f\gets$ source events preceding the target definition and standard background libraries
\State In $E_f$, replace $f$ with a fresh symbol $u_f$ constrained only by $I$
  \Comment{Lemma~\ref{lem:erasure}}
\State $\widehat\varphi\gets$ $\varphi$ with the interface input constraints added
  \Comment{Relativized form in Definition~\ref{def:attribution}}
\If{$\widehat\varphi[u_f/f]$ is proved in $E_f$}
  \State \Return \ding{183} Non-attributable \Comment{Entailed by the interface and background}
\EndIf
\State $c\gets$ proof of $\widehat\varphi$ in $T$; if no proof is obtained, return false for the target or unresolved
\State Construct alternative implementations $A(I)$ and input assignments $X_\varphi$ using the fixed rules
\ForAll{$g\in A(I)$ such that ACL2 accepts $g$ and proves that it satisfies $I$}
  \If{some $\nu$ in $X_\varphi$ satisfies the premises and makes $\widehat\varphi[g/f][\nu]=\bot$}
    \State \Return $(c,(g,\nu))$ \Comment{True for the original target, false for an alternative}
  \EndIf
\EndFor
\If{evaluations are incomplete or no sampled input satisfies the premises}
  \State \Return unresolved
\EndIf
\State \Return \ding{184} No witness found \Comment{Search within scope completed without a witness}
\end{algorithmic}
\end{algorithm}

The alternative implementation in a witness must satisfy the interface; otherwise it would falsify even a property that merely restates the return type. DocAuth builds alternative implementations from the interface rather than from the function body, using two kinds of templates over the function's arguments: one returns a fixed constant, and the other builds the result from the arguments by returning an argument, transforming one, or swapping argument positions. The templates depend only on the function's signature, not on its body, and the witnesses they yield are short enough for readers to inspect directly, as with the constant implementation in Figure~\ref{fig:erasure}. A template instance enters the search only after ACL2 admits it and proves that it satisfies $I$. The system takes the first $16$ instances in a fixed order and keeps those that pass as $A(I)$. The assignments $X_{\varphi}$ are generated from variable types and property premises, and the search stops at the first witness. The full templates, search order, and input-generation rules are in the replication package.

\subsection{Rendering and Rechecking}

The candidate check confirms only $\varphi$, but what is published is the sentence rendered from $\varphi$. Rendering can change a property's meaning or scope, and the sentence need not preserve the implementation-dependent part of $\varphi$, so DocAuth checks the sentence itself. Properties that pass the candidate check are taken in stored order, up to the sentence budget, and each is rendered as one sentence $s$. A separate model call then back-translates $s$ into a formal property $\psi$, given $s$ and the target's definition; because this call never sees $\varphi$, $\psi$ reflects only what the sentence says. The recovered property $\psi$ then goes through the full check of Algorithm~\ref{alg:docauth}: if it is attributable, the sentence is published with $\psi$ and its evidence; otherwise the sentence is discarded. The recheck does not require $\psi$ to equal $\varphi$: readers rely on the sentence, and the evidence delivered with it concerns $\psi$, the property the sentence expresses. The rendering and back-translation prompts are in the replication package.

For the encoder of Section~\ref{sec:attribution}, the property \texttt{(equal (base64-encode-list nil) nil)} can be rendered as ``Encoding an empty character list yields an empty character list.'' When the same property is recovered from this sentence, the recheck again obtains a target proof and the one-character-list witness, and DocAuth delivers the sentence with the recovered property and its evidence. If instead the rendering says only ``The encoded result is a character list,'' the recovered property merely restates the return type and remains provable in the erased environment, so the recheck classifies it as non-attributable and the sentence is not published.

\section{Evaluation}
\label{sec:study}

\subsection{Experimental Setup}
\label{sec:setup}
From the $2{,}054$ functions in $11$ subcorpora of the ACL2 Community Books, we select $162$ functions that have public, fully extractable definitions, take at least one argument, do not access state, support safe renaming and erasure checking, and have at least three proved library theorems. Before generation, we order them pseudo-randomly by a hash of their names and take the first $150$ as targets.

Properties are generated by \texttt{qwen2.5-coder:7b}, \texttt{deepseek-v4-flash}, and \texttt{gpt-5.6-terra}, hereafter Qwen2.5-7B, DeepSeek-v4-flash, and GPT-5.6-terra. All three models use the same targets, prompts, and checking configuration, run at temperature $0$, and are queried once per function for $k=15$ properties. The human-written group consists of the proved theorems that library authors wrote for the same functions, checked under the same configuration.

In the sentence experiment, the three models serve only as sentence writers, and all candidate properties come from the GPT-5.6-terra pool. Each publication path writes at most $4$ sentences per function, and a function with nothing published counts as zero. Back-translation during publication and LLM judging both use GPT-5.6-terra.

Properties were generated on 25--27 August 2026 (UTC); the sentence experiment and timing measurements ran on 10--25 September. Qwen2.5-7B ran locally under Ollama (digest \texttt{dae161e27b0e}); the other models were accessed by model name. On 10 September, DeepSeek replaced V4-Flash with V4.1-Flash, so its properties come from V4-Flash and its sentences from V4.1-Flash, both labelled DeepSeek-v4-flash.

Our main measure is the \emph{attribution rate}, the share of published sentences that a cross-model reassessment confirms as attributable. Throughout, counts of attributable sentences refer to sentences confirmed by this reassessment. The reassessment uses \texttt{claude-sonnet-5}, which plays no part in publication: it recovers formal properties afresh from the opaque definitions and the published sentences and checks them with Algorithm~\ref{alg:docauth}; failed back-translations and unresolved checks count as not confirmed. All 95\% intervals are percentile bootstrap intervals~\citep{efron1993bootstrap} over $10{,}000$ resamples of target functions, which keep each function's observations together~\citep{cameron2008bootstrap}. Differences between rates are computed from the rates as printed, to one decimal place, whereas intervals are computed from unrounded data.

RQ1, RQ2, and RQ5 analyze the properties from the four sources, RQ3 analyzes the sentence experiment, and RQ4 reuses truncated records from both. The five research questions are:
\begin{itemize}[leftmargin=*]
\item \textbf{RQ1 (Prevalence):} How Common Are Non-Attributable Properties?
\item \textbf{RQ2 (Criterion Validity):} Do Attributable Properties Detect Interface-Preserving Implementation Changes?
\item \textbf{RQ3 (Method Effectiveness):} Does DocAuth Improve the Attribution of Generated Documentation?
\item \textbf{RQ4 (Parameter Effects):} How Do Budget Settings Affect the Yield of Attributable Content?
\item \textbf{RQ5 (Coverage and Cost):} How Many Functions Are Covered, and at What Cost?
\end{itemize}

\subsection{RQ1 (Prevalence): How Common Are Non-Attributable Properties?}
All three generator models produce non-attributable properties, and these properties occur in about a third of the functions: among functions with at least one completed check, $35.0\%$, $32.3\%$, and $29.9\%$ have at least one non-attributable property for Qwen2.5-7B, DeepSeek-v4-flash, and GPT-5.6-terra, respectively. Table~\ref{tab:study} gives the check outcomes for each source. Percentages are relative to properties whose check completed; unresolved checks are analyzed in RQ5.

\begin{table}[!htbp]
\centering
\tpTableSetup
\caption{Outcomes of the properties checked on $150$ functions}
\label{tab:study}
\begin{tabular}{@{}lrrrr@{}}
\toprule
 & Qwen2.5-7B & DeepSeek-v4-flash & GPT-5.6-terra & Human-written \\
\midrule
Properties checked & 1{,}743 & 2{,}030 & 2{,}090 & 1{,}068 \\
Checks completed & 939 & 1{,}215 & 1{,}371 & 466 \\
\ding{182}\,Attributable & 461 (49.1\%) & 921 (75.8\%) & 1{,}209 (88.2\%) & 409 (87.8\%) \\
\midrule
False for the target & 373 (39.7\%) & 143 (11.8\%) & 27 (2.0\%) & 0 (0.0\%) \\
\ding{183}\,Non-attributable & 80 (8.5\%) & 77 (6.3\%) & 72 (5.3\%) & 20 (4.3\%) \\
\ding{184}\,No witness found & 25 (2.7\%) & 74 (6.1\%) & 63 (4.6\%) & 37 (7.9\%) \\

\bottomrule
\end{tabular}
\end{table}

Stronger models mainly reduce errors, not non-attribution: from Qwen2.5-7B to GPT-5.6-terra, the share of properties false for the target falls from $39.7\%$ to $2.0\%$, whereas the share of non-attributable ones falls only from $8.5\%$ to $5.3\%$. Identifying non-attributable properties therefore takes a check that goes beyond truth.

Of the proved theorems written by library authors, $87.8\%$ are attributable, on par with GPT-5.6-terra. Nearly all of the $20$ non-attributable ones are lemmas that serve the prover, such as \texttt{(equal (perm (true-list-fix x) y) (perm x y))}, rather than documentation for readers.

\subsection{RQ2 (Criterion Validity): Do Attributable Properties Detect Interface-Preserving Implementation Changes?}
\label{sec:mutation}
Most attributable properties ($64.8\%$) detect interface-preserving implementation changes that played no role in their check, compared with only $15.3\%$ of the properties that also hold for the target but for which no witness was found. Non-attributable properties detect none ($0/48$). This follows from what non-attribution means: such properties are implied by the interface alone, hold for every implementation that preserves it, and carry no information about the implementation. The detection rate of attributable properties, by contrast, is not fixed by the definition.

We apply single-point edits to the function definitions using common mutation-testing operators~\citep{jia2011mutation}: accessor or comparison replacement, argument swaps, constant changes, and conjunct deletion. These edits differ from the witness templates (Section~\ref{sec:audit}). We keep only edits that ACL2 accepts, that change behavior on sampled inputs, and that satisfy the declared return contract; $37$ of the $67$ functions with parseable definitions have at least one usable edit. A property detects an edit if it evaluates to false on some sampled input, and rates are computed over properties with at least one usable edit.

\begin{table}[!htb]
\centering
\tpTableSetup
\caption{Change detection by property source and check outcome}
\label{tab:mutation}
\begin{tabular}{@{}lrrr@{}}
\toprule
Property source & Attributable & No witness found & Non-attributable \\
\midrule
Qwen2.5-7B        &  68/123 (55.3\,\%) & 1/6 (16.7\,\%)  & 0/16 \\
DeepSeek-v4-flash & 204/304 (67.1\,\%) & 4/21 (19.0\,\%) & 0/14 \\
GPT-5.6-terra     & 264/401 (65.8\,\%) & 4/19 (21.1\,\%) & 0/10 \\
Human-written     &  78/120 (65.0\,\%) & 0/13 (0.0\,\%)  & 0/8 \\
\midrule
Total             & 614/948 (64.8\,\%) & 9/59 (15.3\,\%) & 0/48 \\
\bottomrule
\end{tabular}
\end{table}

Pooled over the four groups, $614$ of $948$ attributable properties detect at least one change, with per-group rates ranging from $55.3\%$ to $67.1\%$ (Table~\ref{tab:mutation}). The pooled rate is $49.5$ percentage points above that of no-witness properties (95\% CI $[+31.9,+63.8]$), and the gap is still $47.9$ points on the $18$ functions that have properties of both kinds ($63.7\%$ versus $15.8\%$, 95\% CI $[+25.5,+66.7]$). The $9$ no-witness properties that do detect a change are in fact attributable: the edits that detect them are witnesses outside the search scope. Filtering by attribution sacrifices almost no detection power: of the $507$ property--change failures among GPT-5.6-terra's proved properties, $503$ ($99.2\%$) involve attributable properties.

\FloatBarrier
\subsection{RQ3 (Method Effectiveness): Does DocAuth Improve the Attribution of Generated Documentation?}
\label{sec:sentences}
Pooled over the three generator models, \tool achieves the highest attribution rate of the six methods under all three reassessment conditions and delivers the most attributable sentences. Its edge over truth checking comes mainly from the sentence recheck, while its edge over mutation filtering lies mainly in yield.

We compare six methods on the same functions and generator models, built around three main paths. The \emph{baseline} writes sentences directly from the opaque source code without any check (Figure~\ref{fig:baseline-prompt}). \emph{Truth checking} renders properties proved to hold for the target, and \tool renders attributable properties. These two back-translate every sentence in the same way and differ only in the criterion used for candidate selection and sentence rechecking.

Three variants complete the comparison. \emph{LLM judging} is added after the baseline or truth checking: a judge reads the definition and the sentence and decides whether the sentence describes behavior specific to the target. \emph{Mutation filtering} follows nl2postcond~\citep{endres2024nl2postcond} and MutDafny~\citep{amaral2025mutdafny}: at both candidate selection and sentence rechecking in the truth-checking path, it keeps only properties falsified by at least one mutant. The mutants are all single-point edits of the function body under the RQ2 operators and, like the witness templates, must be proved by ACL2 to satisfy the return contract. Mutation filtering and \tool thus differ mainly in whether alternative implementations come from the function body or from the interface.

\begin{figure}[!htb]
\footnotesize
\hrule\vspace{3pt}
\textbf{Baseline (formal source $\to$ sentence, without intermediate properties).}
You are documenting a function in an ACL2 library so that other developers can understand what it does without reading its implementation.
Write short, factual English sentences describing its behavior. Each sentence should state one point about the function's inputs and result:
the shape of the result, the relationship between the result and inputs, or what happens in boundary cases.
Describe what it computes, not how it is implemented; do not name other library functions or describe internal steps or helper routines.
Output only a numbered list with one sentence per line.
\vspace{3pt}\hrule
\caption{Prompt for direct documentation generation in the baseline}
\Description{The baseline documentation prompt in plain text, requesting short, factual English sentences that describe the function's behavior, stating the shape of the result, the relationship between the result and inputs, or what happens in boundary cases, without mentioning other library functions or internal steps.}
\label{fig:baseline-prompt}
\end{figure}

\paragraph{Overall effect} \tool has the highest cross-model attribution rate of the six methods for every generator model (Table~\ref{tab:independent}, Figure~\ref{fig:sentences}). Pooled over the three models, it reaches $94.7\%$, which is $11.8$ percentage points above truth checking (95\% CI $[+8.5,+15.3]$) and $45.0$ points above the baseline. It publishes fewer sentences than truth checking ($1{,}238$ versus $1{,}301$) yet delivers $94$ more attributable ones ($1{,}172$ versus $1{,}078$, 95\% CI $[+55,+136]$). LLM judging helps little: it raises the baseline by $4.2$ points (95\% CI $[+2.9,+5.5]$) and truth checking by a non-significant $1.6$ points (95\% CI $[-0.1,+3.8]$).

\begin{table}[!b]
\centering
\tpTableSetup
\setlength{\tabcolsep}{3.7pt}
\caption{Attributable sentences and attribution rates by generator model and reassessment condition}
\label{tab:independent}
\begin{tabular}{@{}llrrrrr@{}}
\toprule
Generator model & Method & Published & Attributable & Cross-model & Interface-only & \shortstack{Interface-only,\\equivalent to $\psi$} \\
\midrule
\multirow{6}{*}{Qwen2.5-7B} & Baseline & 575 & 247 & 43.0\% & 26.1\% & 20.2\% \\
 & \quad +\,LLM judge & 404 & 217 & 53.7\% & 34.2\% & 26.5\% \\
 & Truth check & 376 & 299 & 79.5\% & 62.5\% & 43.9\% \\
 & \quad +\,LLM judge & 359 & 290 & 80.8\% & 63.5\% & 44.8\% \\
 & \quad +\,mutation filter & 265 & 229 & 86.4\% & 64.5\% & 44.2\% \\
 & \tool & 359 & \textbf{317} & \textbf{88.3\%} & \textbf{68.0\%} & \textbf{46.8\%} \\
\midrule
\multirow{6}{*}{DeepSeek-v4-flash} & Baseline & 577 & 303 & 52.5\% & 37.1\% & 28.9\% \\
 & \quad +\,LLM judge & 454 & 253 & 55.7\% & 38.3\% & 29.1\% \\
 & Truth check & 458 & 384 & 83.8\% & 67.5\% & 48.3\% \\
 & \quad +\,LLM judge & 382 & 328 & 85.9\% & 66.8\% & 44.8\% \\
 & \quad +\,mutation filter & 320 & 303 & 94.7\% & 75.0\% & \textbf{54.4\%} \\
 & \tool & 432 & \textbf{417} & \textbf{96.5\%} & \textbf{77.3\%} & 52.8\% \\
\midrule
\multirow{6}{*}{GPT-5.6-terra} & Baseline & 572 & 307 & 53.7\% & 39.2\% & 29.2\% \\
 & \quad +\,LLM judge & 511 & 268 & 52.4\% & 36.6\% & 25.6\% \\
 & Truth check & 467 & 395 & 84.6\% & 68.5\% & 46.3\% \\
 & \quad +\,LLM judge & 372 & 323 & 86.8\% & 67.5\% & 40.9\% \\
 & \quad +\,mutation filter & 342 & 321 & 93.9\% & 75.1\% & 48.8\% \\
 & \tool & 447 & \textbf{438} & \textbf{98.0\%} & \textbf{80.3\%} & \textbf{57.3\%} \\
\bottomrule
\end{tabular}
\end{table}

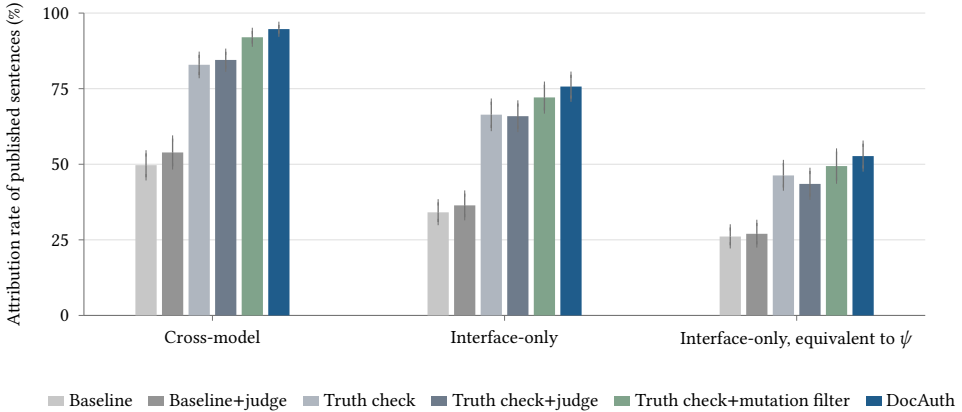
\begin{figure}[!htb]
\centering
\begin{tikzpicture}
\begin{axis}[
  tpFig, scale only axis, width=0.80\textwidth, height=4.0cm,
  ybar, bar width=8pt,
  ymin=0, ymax=100, ytick={0,25,50,75,100},
  ylabel={Attribution rate of published sentences (\%)},
  symbolic x coords={c1,c2,c3}, xtick=data,
  xticklabels={Cross-model,Interface-only,{Interface-only, equivalent to $\psi$}},
  enlarge x limits=0.22,
  axis x line*=bottom,
  legend style={legend columns=-1, at={(0.5,-0.22)}, anchor=north,
                /tikz/every even column/.append style={column sep=4pt}},
  legend image code/.code={\draw[#1,draw=none] (0cm,-0.085cm) rectangle (0.30cm,0.105cm);},
  error bars/error bar style={draw=black!55, line width=.5pt},
  error bars/error mark options={mark size=1.2pt, draw=black!55},
]
\addplot[fill=black!22,draw=none,error bars/.cd,y dir=both,y explicit] coordinates
  {(c1,49.7) += (0,3.9) -= (0,4.0) (c2,34.1) += (0,3.3) -= (0,3.2) (c3,26.1) += (0,3.0) -= (0,2.9)};
\addlegendentry{Baseline}
\addplot[fill=black!42,draw=none,error bars/.cd,y dir=both,y explicit] coordinates
  {(c1,53.9) += (0,4.6) -= (0,4.6) (c2,36.4) += (0,3.9) -= (0,3.9) (c3,27.0) += (0,3.6) -= (0,3.5)};
\addlegendentry{Baseline$+$judge}
\addplot[fill=tpHuman!55,draw=none,error bars/.cd,y dir=both,y explicit] coordinates
  {(c1,82.9) += (0,3.3) -= (0,3.4) (c2,66.4) += (0,4.3) -= (0,4.4) (c3,46.3) += (0,4.1) -= (0,4.0)};
\addlegendentry{Truth check}
\addplot[fill=tpHuman,draw=none,error bars/.cd,y dir=both,y explicit] coordinates
  {(c1,84.5) += (0,2.7) -= (0,2.9) (c2,65.9) += (0,4.2) -= (0,4.2) (c3,43.5) += (0,4.3) -= (0,4.2)};
\addlegendentry{Truth check$+$judge}
\addplot[fill=tpOk!70,draw=none,error bars/.cd,y dir=both,y explicit] coordinates
  {(c1,92.0) += (0,2.1) -= (0,2.1) (c2,72.1) += (0,4.2) -= (0,4.3) (c3,49.4) += (0,4.8) -= (0,4.8)};
\addlegendentry{Truth check$+$mutation filter}
\addplot[fill=tpAttr,draw=none,error bars/.cd,y dir=both,y explicit] coordinates
  {(c1,94.7) += (0,1.4) -= (0,1.5) (c2,75.7) += (0,3.9) -= (0,4.0) (c3,52.7) += (0,4.1) -= (0,4.1)};
\addlegendentry{\tool}
\end{axis}
\end{tikzpicture}
\caption{Sentence attribution rates, pooled over the three models}
\Description{Grouped bar chart. The horizontal axis shows the three reassessment conditions: cross-model, interface-only, and interface-only with equivalence to psi. Each group has six bars, for the baseline, the baseline with LLM judging, truth checking, truth checking with LLM judging, truth checking with mutation filtering, and \tool; error bars are 95\% bootstrap intervals clustered by target function. \tool's bar is the tallest in all three groups, at 94.7\%, 75.7\%, and 52.7\%, respectively. Table~\ref{tab:independent} gives the values for each generator model.}
\label{fig:sentences}
\end{figure}

\paragraph{Roles of the two checks} The sentence recheck raises the attribution rate, whereas the candidate check raises the yield. Switching only the recheck criterion of the truth-checking path to attribution lifts its rate from $82.9\%$ to $93.9\%$ ($+11.0$ points, 95\% CI $[+7.9,+14.4]$); $7.4$ points of this gain (95\% CI $[+4.5,+10.6]$) come from removing non-attributable and no-witness sentences. The same recheck also lifts the baseline from $49.7\%$ to $86.7\%$. Under the same recheck, \tool and this truth-checking path have similar rates ($94.7\%$ versus $93.9\%$, 95\% CI $[-0.8,+2.3]$), but \tool delivers $103$ more attributable sentences (95\% CI $[+63,+145]$): sentences written from attributable properties pass the recheck more often ($87.7\%$, versus $78.4\%$ for truth checking and $50.1\%$ for the baseline).

\paragraph{Comparison with mutation filtering} Mutation filtering comes close to the attribution rate of \tool but delivers far fewer attributable sentences. Its rate is $92.0\%$, $9.1$ points above truth checking (95\% CI $[+5.8,+12.7]$), and \tool is another $2.7$ points higher (95\% CI $[+0.6,+4.8]$). Mutation filtering, however, publishes only $927$ sentences, $853$ of them attributable, so \tool delivers $319$ more (95\% CI $[+236,+404]$). The gap arises in candidate selection. The mutants falsify no non-attributable property, but they also falsify only $60.4\%$ of the attributable ones, so mutation filtering selects $374$ properties (\tool: $481$) and none at all for $52$ functions (\tool: $27$). A single-point edit often leaves the behavior a property describes untouched, whereas alternative implementations built from the interface are not limited by the function body.

\paragraph{Per-function delivery} Each of the three generator models documents each function once, giving $450$ documents. \tool delivers at least one attributable sentence in $80.9\%$ of them, on par with the baseline's $83.3\%$ (95\% CI of the difference $[-7.1,+2.0]$). It delivers four attributable sentences in $40.9\%$ of documents, versus $8.0\%$ for the baseline, and $2.60$ per document on average, compared with $2.40$ for truth checking and $1.90$ for the baseline (Table~\ref{tab:delivery}). The $18.7\%$ of documents in which \tool publishes nothing stem mainly from the $27$ functions whose property pool contains no attributable property.

\begin{table}[!htb]
\centering
\tpTableSetup
\caption{Per-function delivery}
\label{tab:delivery}
\begin{tabular}{@{}lrrrrrr@{}}
\toprule
Method & \shortstack{Nothing\\published} & \shortstack{${\geq}1$\\attributable} & \shortstack{${\geq}2$\\attributable} & \shortstack{$4$\\attributable} & \shortstack{Attributable\\per document} & \shortstack{Attribution\\rate} \\
\midrule
Baseline & 0.7\% & 83.3\% & 65.8\% & 8.0\% & 1.90 & 49.7\% \\
Truth check & 16.9\% & 79.3\% & 72.0\% & 30.9\% & 2.40 & 82.9\% \\
\quad +\,mutation filter & 36.4\% & 62.4\% & 55.6\% & 26.4\% & 1.90 & 92.0\% \\
\tool & 18.7\% & 80.9\% & 74.7\% & 40.9\% & 2.60 & 94.7\% \\
\bottomrule
\end{tabular}
\end{table}

\paragraph{Composition of published content} The attribution criterion tilts published content only slightly toward point and boundary facts, and this shift explains almost none of the rate gain. We classify each published sentence, by the syntax of the $\psi$ delivered with it, as a point fact, a boundary case (an argument takes a base value, or a premise restricts it to its base case), a type or shape fact, or a general relation (rules in the replication package). Point and boundary facts make up $38.4\%$ of \tool's sentences and $35.1\%$ of truth checking's (95\% CI of the difference $[+1.3,+5.4]$). Both exceed the $20.0\%$ in the property pool, because both paths take properties from the front of the list, where models tend to write base cases. The $3.3$-point difference comes mainly from incomplete checking: among properties proved true of the target, $15.0\%$ of general relations remain unresolved or have no witness found, versus $1.7\%$ of point and boundary facts; this is the same limitation that bounds coverage in RQ5. A Kitagawa decomposition (no interaction residue) assigns only $0.4$ of \tool's $11.8$-point lead over truth checking to the class mix, and on general relations alone \tool is still $16.3$ points higher ($90.1\%$ versus $73.8\%$, 95\% CI $[+11.1,+21.5]$). The classes do not measure importance, but attributable boundary facts are not trivial either: in RQ2 they detect changes about as often as general relations do ($60.9\%$ versus $68.9\%$).

\paragraph{Reassessment without code} The cross-model reassessment reads the function body and may use it to fill in what a sentence leaves unsaid. We therefore reassess all published sentences under two stricter conditions. In the first, the reassessment model sees only the interface (function name, formal parameters, input constraint, and return contract); the second additionally requires the recovered property to be proved equivalent in ACL2 to the $\psi$ delivered with the sentence. Pooled over the three models, \tool remains the highest of the six methods ($75.7\%$ and $52.7\%$), $9.3$ and $6.4$ points above truth checking (95\% CIs $[+6.0,+12.7]$ and $[+3.4,+9.6]$). Its rate differences from mutation filtering are not significant ($+3.6$ and $+3.3$ points, 95\% CIs $[-0.1,+7.3]$ and $[-1.2,+7.6]$), but \tool delivers $269$ and $194$ more attributable sentences (95\% CIs $[+198,+341]$ and $[+130,+258]$). Without the code, $20.5\%$ of the sentences originally confirmed for \tool and $21.2\%$ of those for truth checking are no longer confirmed. The losses are similar, so \tool's advantage does not depend on the reassessment model reading the code.

\paragraph{Sensitivity and stability of the recheck} The sentence recheck is sensitive to changes in a sentence's meaning, and its results are stable. A model that sees only the sentence rewrote $230$ confirmed-attributable sentences, stratified across the three paths. When a rewrite changes one result, premise, or constant, $70.0\%$ of the rewritten sentences fail the recheck and $45.2\%$ yield a property that is false for the target; when it changes only the wording, $86.5\%$ still pass. Of $337$ sentences reassessed twice, $94.1\%$ (\tool) and $91.3\%$ (truth checking) yield equivalent properties in both runs.

\subsection{RQ4 (Parameter Effects): How Do Budget Settings Affect the Yield of Attributable Content?}
\label{sec:budget}
Larger candidate budgets and search scopes yield more attributable content, with diminishing returns, and writing more sentences per function barely lowers the attribution rate. Instead of regenerating, we truncate the saved records to vary the number of candidates per function $k$ (Figure~\ref{fig:budget}a), the witness-search scope (Figure~\ref{fig:budget}b), and the number of sentences per function $m$ (Figure~\ref{fig:budget}c,d).

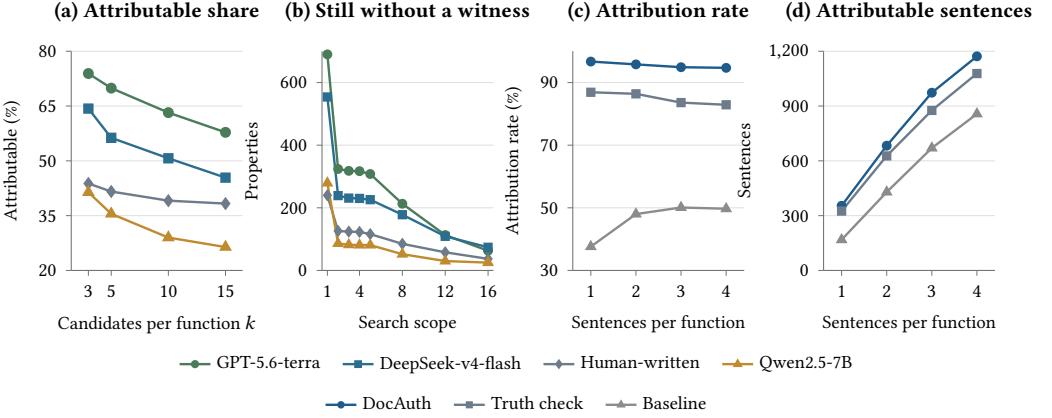
\begin{figure}[!htb]
\centering
\begin{tikzpicture}
\begin{groupplot}[
  tpFig, scale only axis,
  group style={group size=4 by 1, horizontal sep=1.05cm},
  width=0.163\textwidth, height=2.9cm,
  axis x line*=bottom,
  legend style={legend columns=4, draw=none, fill=none,
                /tikz/every even column/.append style={column sep=6pt}},
]
\nextgroupplot[xmin=1.5, xmax=16.5, xtick={3,5,10,15}, xlabel={Candidates per function $k$},
  ymin=20, ymax=80, ytick={20,35,50,65,80},
  ylabel={Attributable (\%)}, title={(a) Attributable share}, legend to name=budgetlegend]
\addplot[color=tpOk,mark=*,mark size=1.6pt,line width=.9pt] coordinates {(3,73.9) (5,69.9) (10,63.2) (15,57.8)};
\addlegendentry{GPT-5.6-terra}
\addplot[color=tpDS,mark=square*,mark size=1.5pt,line width=.9pt] coordinates {(3,64.3) (5,56.3) (10,50.7) (15,45.4)};
\addlegendentry{DeepSeek-v4-flash}
\addplot[color=tpHuman,mark=diamond*,mark size=1.8pt,line width=.9pt] coordinates {(3,43.8) (5,41.6) (10,39.1) (15,38.3)};
\addlegendentry{Human-written}
\addplot[color=tpQW,mark=triangle*,mark size=1.9pt,line width=.9pt] coordinates {(3,41.4) (5,35.5) (10,29.0) (15,26.4)};
\addlegendentry{Qwen2.5-7B}
\nextgroupplot[xmin=0.5, xmax=16.5, xtick={1,4,8,12,16},
  ymin=0, ymax=700, ytick={0,200,400,600},
  xlabel={Search scope}, ylabel={Properties}, title={(b) Still without a witness}]
\addplot[color=tpOk,mark=*,mark size=1.4pt,line width=.9pt] coordinates {(1,690) (2,324) (3,318) (4,317) (5,308) (8,213) (12,113) (16,63)};
\addplot[color=tpDS,mark=square*,mark size=1.3pt,line width=.9pt] coordinates {(1,554) (2,239) (3,231) (4,230) (5,226) (8,178) (12,109) (16,74)};
\addplot[color=tpHuman,mark=diamond*,mark size=1.6pt,line width=.9pt] coordinates {(1,240) (2,126) (3,124) (4,123) (5,116) (8,85) (12,58) (16,37)};
\addplot[color=tpQW,mark=triangle*,mark size=1.7pt,line width=.9pt] coordinates {(1,279) (2,86) (3,82) (4,81) (5,81) (8,52) (12,30) (16,25)};
\nextgroupplot[xmin=0.6, xmax=4.4, xtick={1,2,3,4},
  ymin=30, ymax=100, ytick={30,50,70,90},
  xlabel={Sentences per function}, ylabel={Attribution rate (\%)}, title={(c) Attribution rate}, legend to name=armlegend]
\addplot[color=tpAttr,mark=*,mark size=1.4pt,line width=.9pt] coordinates {(1,96.7) (2,95.8) (3,94.9) (4,94.7)};
\addlegendentry{\tool}
\addplot[color=tpHuman,mark=square*,mark size=1.3pt,line width=.9pt] coordinates {(1,86.9) (2,86.4) (3,83.6) (4,82.9)};
\addlegendentry{Truth check}
\addplot[color=black!45,mark=triangle*,mark size=1.7pt,line width=.9pt] coordinates {(1,37.6) (2,48.0) (3,50.1) (4,49.7)};
\addlegendentry{Baseline}
\nextgroupplot[xmin=0.6, xmax=4.4, xtick={1,2,3,4},
  ymin=0, ymax=1200, ytick={0,300,600,900,1200},
  xlabel={Sentences per function}, ylabel={Sentences}, title={(d) Attributable sentences}]
\addplot[color=tpAttr,mark=*,mark size=1.4pt,line width=.9pt] coordinates {(1,354) (2,683) (3,973) (4,1172)};
\addplot[color=tpHuman,mark=square*,mark size=1.3pt,line width=.9pt] coordinates {(1,325) (2,627) (3,876) (4,1078)};
\addplot[color=black!45,mark=triangle*,mark size=1.7pt,line width=.9pt] coordinates {(1,168) (2,429) (3,670) (4,857)};
\end{groupplot}
\end{tikzpicture}
\pgfplotslegendfromname{budgetlegend}
\pgfplotslegendfromname{armlegend}
\caption{Effects of the candidate budget, the search scope, and the number of sentences per function}
\Description{Four line plots in one row. (a) As the candidate budget grows, the share of attributable properties falls for all four property sources. (b) The number of properties still without a witness when the search is limited to a prefix of the alternative implementations drops sharply at the second alternative and keeps falling to the no-witness-found count of the full search. (c, d) With published sentences truncated per function, DocAuth's attribution rate stays above that of truth checking, and the number of attributable sentences delivered by all three methods grows with the number of sentences per function.}
\label{fig:budget}
\end{figure}

\paragraph{Candidate budget} When each function's candidates are truncated to the first $k\in\{3,5,10,15\}$, the number of attributable properties grows with $k$ while their share falls (Figure~\ref{fig:budget}a). For GPT-5.6-terra, the count rises from $311$ to $1{,}209$ and the share drops from $73.9\%$ to $57.8\%$. The other three groups follow the same trend; candidates later in the list more often turn out false for the target or unresolved. The change-detection rate of attributable properties stays between $55\%$ and $70\%$ at every $k$. The number of functions with enough attributable properties for four sentences saturates quickly: for GPT-5.6-terra, it is $113$ at $k=10$, and $k=15$ adds only $5$ more.

\paragraph{Search scope} The witness search tries alternative implementations in a fixed order and stops at the first witness. The number of properties still without a witness drops most sharply at the second alternative, from $690$ to $324$ for GPT-5.6-terra; after the $16$th, $63$ remain, which is the no-witness count in Table~\ref{tab:study} (Figure~\ref{fig:budget}b). The curves are still falling at that point, so a larger scope would find witnesses for some of these properties too.

\paragraph{Sentences per function} When each function's published sentences are truncated to the first $m\in\{1,2,3,4\}$, \tool's attribution rate falls only from $96.7\%$ to $94.7\%$, while its attributable sentences grow from $354$ to $1{,}172$. At every $m$, its rate is $9.4$--$11.8$ points above that of truth checking (Figure~\ref{fig:budget}c,d).

\subsection{RQ5 (Coverage and Cost): How Many Functions Are Covered, and at What Cost?}
\label{sec:instrument}
The four property sources cover $70.0\%$--$84.7\%$ of the functions, and coverage is limited mainly by unresolved checks rather than by the witness search. Running \tool end to end takes a median of $219.8$ seconds per function, of which only $14.1$ seconds are spent in the prover, and costs about $\$0.048$ in model calls.

\paragraph{Function coverage} A function counts as covered if it obtains at least one attributable property, and the denominator is all $150$ targets (Table~\ref{tab:coverage}). Coverage depends on how attributable properties are distributed across functions, not only on their number: GPT-5.6-terra obtains more attributable properties than DeepSeek-v4-flash ($1{,}209$ versus $921$) but covers fewer functions ($123$ versus $127$).

\begin{table}[!htb]
\centering
\tpTableSetup
\caption{Function coverage by property source}
\label{tab:coverage}
\begin{tabular}{@{}lrrr@{}}
\toprule
Property source & Properties checked & Functions covered & Coverage \\
\midrule
Qwen2.5-7B & 1{,}743 & 114/150 & 76.0\% \\
DeepSeek-v4-flash & 2{,}030 & 127/150 & 84.7\% \\
GPT-5.6-terra & 2{,}090 & 123/150 & 82.0\% \\
Human-written & 1{,}068 & 105/150 & 70.0\% \\
\bottomrule
\end{tabular}
\end{table}

\begin{figure}[!htb]
\centering
\begin{tikzpicture}
\begin{groupplot}[
  tpFig, scale only axis, group style={group size=2 by 1, horizontal sep=1.6cm},
  width=0.40\textwidth, height=4.2cm,
  ymin=0, ymax=100, ytick={0,25,50,75,100},
  axis x line*=bottom,
  ybar stacked, /pgf/bar width=15pt,
  symbolic x coords={Qwen,DeepSeek,GPT,Human}, xtick=data,
  xticklabels={{Qwen2.5-\\7B},{DeepSeek-\\v4-flash},{GPT-5.6-\\terra},{Human-\\written}},
  xticklabel style={align=center},
  enlarge x limits=0.18,
  legend image code/.code={\draw[#1,draw=none] (0cm,-0.085cm) rectangle (0.30cm,0.105cm);},
  legend style={legend columns=1, at={(0.02,-0.20)}, anchor=north west,
                cells={anchor=west}, font=\scriptsize, draw=none},
  /tikz/rq5dark/.style={font=\tiny, text=white},
  /tikz/rq5light/.style={font=\tiny, text=black!80},
]
\nextgroupplot[ylabel={Share of checked properties (\%)}, title={(a) Check outcomes}, legend columns=3,
  after end axis/.code={
  \node[rq5dark] at (axis cs:Qwen,13.2) {26.4};
  \node[rq5dark] at (axis cs:Qwen,37.1) {21.4};
  \node[rq5dark] at (axis cs:DeepSeek,22.7) {45.4};
  \node[rq5dark] at (axis cs:GPT,28.9) {57.8};
  \node[rq5dark] at (axis cs:Human,19.1) {38.3};
  \node[rq5light] at (axis cs:Qwen,76.9) {46.1};
  \node[rq5light] at (axis cs:DeepSeek,79.9) {40.1};
  \node[rq5light] at (axis cs:GPT,82.8) {34.4};
  \node[rq5light] at (axis cs:Human,71.8) {56.4};}]
\addplot[fill=tpAttr,draw=white,line width=.3pt] coordinates {(Qwen,26.4) (DeepSeek,45.4) (GPT,57.8) (Human,38.3)};
\addlegendentry{Attributable}
\addplot[fill=tpHuman,draw=white,line width=.3pt] coordinates {(Qwen,21.4) (DeepSeek,7.0) (GPT,1.3) (Human,0.0)};
\addlegendentry{False for the target}
\addplot[fill=tpHuman!45,draw=white,line width=.3pt] coordinates {(Qwen,4.6) (DeepSeek,3.8) (GPT,3.4) (Human,1.9)};
\addlegendentry{Non-attributable}
\addplot[fill=black!60,draw=white,line width=.3pt] coordinates {(Qwen,1.4) (DeepSeek,3.6) (GPT,3.0) (Human,3.5)};
\addlegendentry{No witness found}
\addplot[fill=black!10,draw=white,line width=.3pt] coordinates {(Qwen,46.1) (DeepSeek,40.1) (GPT,34.4) (Human,56.4)};
\addlegendentry{Unresolved}
\nextgroupplot[ylabel={Share of unresolved properties (\%)}, title={(b) Reasons for unresolved outcomes},
  after end axis/.code={
  \node[rq5dark] at (axis cs:Qwen,12.8) {25.6};
  \node[rq5dark] at (axis cs:DeepSeek,16.9) {33.7};
  \node[rq5dark] at (axis cs:GPT,20.2) {40.3};
  \node[rq5dark] at (axis cs:Human,15.4) {30.9};
  \node[rq5dark] at (axis cs:Qwen,37.4) {23.5};
  \node[rq5dark] at (axis cs:DeepSeek,44.1) {20.6};
  \node[rq5dark] at (axis cs:GPT,50.5) {20.3};
  \node[rq5dark] at (axis cs:Human,52.4) {43.0};
  \node[rq5light] at (axis cs:Qwen,68.2) {38.3};
  \node[rq5light] at (axis cs:DeepSeek,71.0) {33.4};
  \node[rq5light] at (axis cs:GPT,74.8) {28.5};
  \node[rq5light] at (axis cs:Human,82.6) {17.3};
  \node[rq5light] at (axis cs:Qwen,93.7) {12.6};
  \node[rq5light] at (axis cs:DeepSeek,93.9) {12.3};
  \node[rq5light] at (axis cs:GPT,94.6) {10.8};
  \node[rq5light] at (axis cs:Human,95.6) {8.8};}]
\addplot[fill=black!70,draw=white,line width=.3pt] coordinates {(Qwen,25.6) (DeepSeek,33.7) (GPT,40.3) (Human,30.9)};
\addlegendentry{No target proof}
\addplot[fill=black!50,draw=white,line width=.3pt] coordinates {(Qwen,23.5) (DeepSeek,20.6) (GPT,20.3) (Human,43.0)};
\addlegendentry{Missing dependency or erasure problem}
\addplot[fill=black!30,draw=white,line width=.3pt] coordinates {(Qwen,38.3) (DeepSeek,33.4) (GPT,28.5) (Human,17.3)};
\addlegendentry{Evaluation failure or unmet premises}
\addplot[fill=black!12,draw=white,line width=.3pt] coordinates {(Qwen,12.6) (DeepSeek,12.3) (GPT,10.8) (Human,8.8)};
\addlegendentry{Parsing, timeout, or other}
\end{groupplot}
\end{tikzpicture}
\caption{Check outcomes and reasons for unresolved outcomes by property source}
\Description{Each panel has four stacked percentage bars for Qwen2.5-7B, DeepSeek-v4-flash, GPT-5.6-terra, and human-written properties. Panel (a) partitions all checked properties into attributable, false for the target, non-attributable, no witness found, and unresolved. Panel (b) partitions each source's unresolved properties into missing target proofs, missing dependencies or erasure-environment problems, evaluation failures or unmet premises, and other incomplete checks.}
\label{fig:rq5-outcomes}
\end{figure}
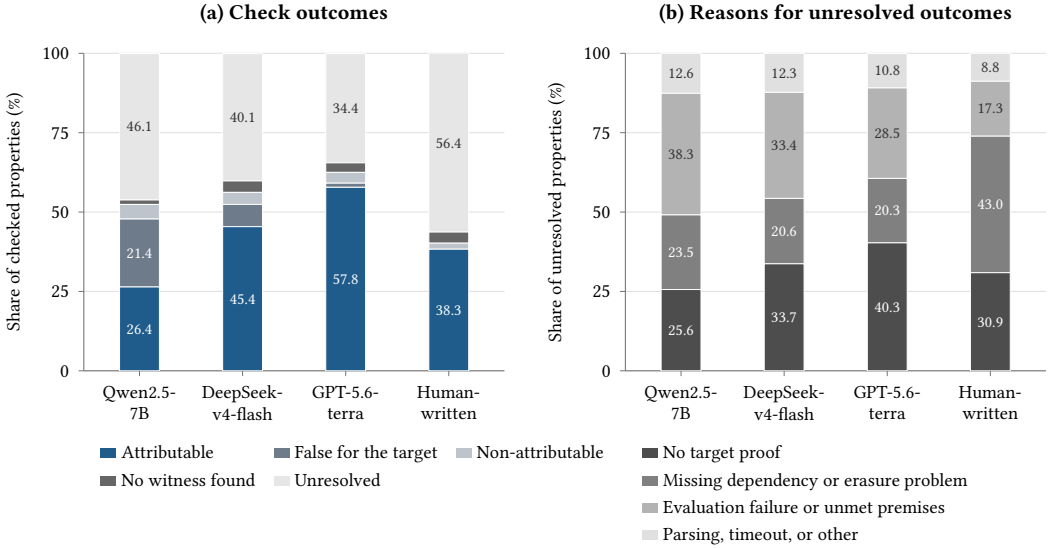

\begin{figure}[!htb]
\centering
\includegraphics[width=\textwidth]{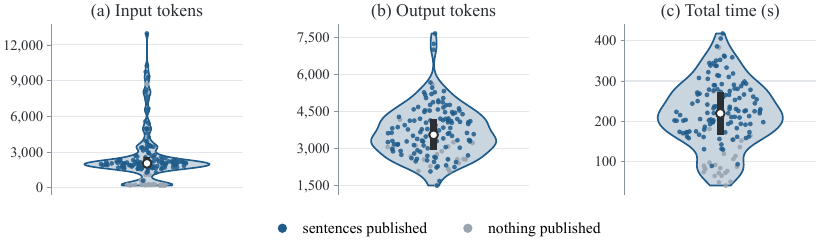}
\caption{Distribution of token usage and total time per function}
\Description{Three violin plots show the distributions of input tokens, output tokens, and total time over the 149 functions. Each plot has a quartile bar and a median dot in the middle, and each function is drawn as a point: dark blue for functions that publish sentences and gray for those that publish nothing. Functions that publish nothing cluster at the low end of total time and also use the fewest input tokens.}
\label{fig:cost}
\end{figure}

\paragraph{What limits coverage} Unresolved outcomes account for $34.4\%$--$56.4\%$ of all checked properties, far more than the $1.4\%$--$3.6\%$ with no witness found (Figure~\ref{fig:rq5-outcomes}a). For generated properties, most unresolved outcomes stem from missing target proofs or evaluation failures (Figure~\ref{fig:rq5-outcomes}b). Human-written theorems have the highest unresolved share, and $43.0\%$ of their unresolved outcomes are due to missing dependencies or erasure-environment problems: these theorems often involve the target indirectly, through other library functions that call it, and such properties are left undecided in the erased environment (Section~\ref{sec:audit-outcomes}).

\paragraph{Runtime cost}
\label{sec:cost}
We time an end-to-end run of \tool from fresh candidates with GPT-5.6-terra on all $150$ functions ($15$ candidates per function, a single process on an Intel i7-14700HX). The model service's content filter rejects one function's definition; the remaining $149$ complete the measurement (Figure~\ref{fig:cost}). The median total time per function is $219.8$ seconds ($p_{10}$--$p_{90}$: $91.8$--$324.1$), of which $197.5$ seconds are spent on model calls and $14.1$ seconds in the prover. The median input and output token counts are $2{,}054$ and $3{,}551$. At the list price of GPT-5.6-terra ($\$2$ and $\$12$ per million input and output tokens), the median model-call cost is about $\$0.048$ per function, or about $\$7.3$ in total for the $149$ functions. The $27$ functions with no published sentence stop after candidate checking, with a median total time of only $90.1$ seconds. Timings depend on API response speed, and measurements taken at different times can vary considerably.

\section{Discussion}
\label{sec:discussion}

Attribution fills a gap in how documentation is evaluated. In software engineering, a test is judged not only by whether it passes but also by whether it fails on faulty implementations~\citep{jia2011mutation}. Generated documentation and specifications, by contrast, are mostly judged by the first criterion alone: whether they can be proved or are consistent with the code~\citep{ma2025specgen,sun2023clover}. For statements that describe a function's behavior, attribution imposes the second criterion: a true statement should also fail for at least one implementation that satisfies the same interface. Since stronger models mainly reduce errors rather than non-attribution (RQ1), and LLM judging helps little (RQ3), evaluations of documentation- and specification-generation methods can report the attribution rate alongside correctness measures.

Unlike scores based on textual similarity to reference comments or on human or LLM ratings, every attribution verdict comes with checkable evidence: a target proof and a witness confirm attribution, and an erasure proof establishes non-attribution; for a sentence, the verdict concerns the property $\psi$ recovered from it. Because the measure is relative to the library's published interface, a stronger interface leaves fewer statements attributable, so attribution rates of different methods should be compared under the same interface. The measure also applies to existing documentation: when it is drawn from library theorems, lemmas written for the prover can be set aside first (RQ1).

We see three directions for future work. First, attribution can be extended to documentation comments in mainstream languages. There, executable alternative implementations that satisfy a signature or contract can supply witnesses, whereas proving that a statement holds for the target and establishing non-attribution still require a verifier. Second, attribution can be combined with measures of importance, completeness, and readability into an overall assessment of documentation quality, and developer studies can test how this assessment relates to the usefulness of documentation. Third, studies of real code evolution can test whether attributable statements also serve as regression checks; RQ2 already provides preliminary evidence for single-point edits.

\section{Threats to Validity}
\label{sec:threats}

\textbf{Construct validity and developer feedback.} The attribution rate measures whether statements meet the basic requirement of describing behavior beyond the interface. Whether documentation is useful to developers also depends on importance, completeness, and readability, which need their own measures. We did not conduct a developer study, so it remains untested whether attributable documentation helps developers understand and maintain code, and how much effort it takes to inspect the attached evidence.

\textbf{Measurement.} The formal meaning of each sentence is recovered by an LLM, and the reassessment uses the same checker as publication. The interface-only reassessment and the meaning-change experiment in RQ3 reduce the risk that recovery relies on the code, but they do not replace human judgment of what a sentence means.

\textbf{Experimental scope.} The experiments cover stateless functions with extractable definitions from the ACL2 Community Books and three LLMs, each run once at temperature $0$; other prompts, sampling settings, or model versions may yield different candidate properties. Programs with mutable state, external interaction, or concurrency place different demands on interface modeling and alternative implementations. Other provers need their own interface extraction and erasure environment; there, the natural counterpart is a bodyless function constrained by its declared specification.

\section{Related Work}
\label{sec:related}
\paragraph{Can specifications detect incorrect implementations?}
nl2postcond~\citep{endres2024nl2postcond} generates postconditions from natural-language requirements and checks both whether they are correct and whether they distinguish correct from incorrect code. MutDafny~\citep{amaral2025mutdafny} mutates code to find specifications too weak to notice the change, and IronSpec~\citep{goldweber2024ironspec} tests whether specifications match developer intent. Like these methods, we change the implementation and observe whether a property still holds. Unlike them, we measure against the public interface, judge individual statements already known to be true, and can prove non-attribution by erasure instead of merely failing to find a distinguishing implementation. Vacuity checking and semantic dependence likewise change one part of an artifact to decide whether a conclusion depends on it~\citep{beer2001vacuity,podgurski1990dependence,chockler2008causes}; we instead change the whole function's behavior while keeping its published interface.

\paragraph{Does documentation describe code accurately?}
LLMs are already used to generate code comments and summaries~\citep{geng2024fewshot,sun2023chatgpt}, which are mostly evaluated by textual similarity to reference comments or by the correctness of their content, that is, by whether the documentation says anything wrong. \citet{panthaplackel2021deep} detect comments that become outdated as code changes, Clover~\citep{sun2023clover} checks consistency among code, documentation, and formal annotations, and CASCADE~\citep{kiecker2026cascade} generates tests from documentation and runs them on the code to find documentation that contradicts it. \emph{Attribution} in natural language generation is a different notion: it requires generated text to be supported by an identifiable source~\citep{rashkin2023measuring}. Truth, source support, and consistency with the code do not rule out the sentences in Figure~\ref{fig:motivating}, which are consistent with the code and well supported yet only restate the interface or a general fact; this is precisely what our notion of attribution checks. Documentary~\citep{eghbali2026documentary} and VeriEquivBench~\citep{zeng2025veriequivbench} check whether an entire document or specification is behaviorally equivalent to the code. Their unit of decision is overall behavior, whereas we examine the behavioral information that a single statement provides. \citet{khatib2026mutation} change the code, regenerate summaries, and observe whether the summaries change, which evaluates the summarization model itself; we do not regenerate, but check whether each written sentence still holds for alternative implementations.

\paragraph{Property generation and documentation maintenance}
Tools such as QuickSpec~\citep{claessen2010quickspec}, Daikon~\citep{ernst2007daikon}, and SpecGen~\citep{ma2025specgen} automatically discover or generate function properties, and some of them prune redundant results. QuickSpec does not report equations that follow from ones it has already found; IsaCoSy~\citep{johansson2011conjecture}, HipSpec~\citep{claessen2013hipspec}, and Hipster~\citep{johansson2014hipster} discard conjectures that follow from known lemmas during theory exploration; and Daikon suppresses invariants implied by other invariants~\citep{ernst2000quickly}. Such pruning is relative to the other properties already obtained, and a pruned property may still describe implementation-specific behavior. Attribution, by contrast, is relative to the library's published interface. It judges each true statement on its own, asking whether it says anything about behavior beyond the interface, and it confirms attribution with a witness and establishes non-attribution with an erasure proof, providing a basis for selecting documentation content. DocAuth's recheck recovers formal properties from natural language and is therefore related to autoformalization~\citep{wu2022autoformalization,cosler2023nl2spec}; rather than aiming for a faithful sentence-by-sentence translation, however, we check whether the recovered property is still attributable. XDoc~\citep{davis2014industrial} organizes and publishes ACL2 library documentation, which complements our content selection.

\Needspace{5\baselineskip}
\section{Conclusion}

A provably true documentation statement may still do no more than restate the interface or a general fact. This paper turns that distinction into a checkable condition: relative to the public interface, a statement that holds for the target is attributable exactly when it fails for some alternative implementation that satisfies the same interface, and an erasure proof establishes when no such implementation exists. Non-attributable properties occur in about a third of the functions we studied, whereas most attributable properties detect interface-preserving implementation changes. By adding attribution on top of truth checking, DocAuth delivers more attributable sentences than both truth checking and mutation filtering, at a higher attribution rate, and publishes each sentence with the formal property recovered from it and that property's evidence. Just as a test should not only pass but also fail on faulty implementations, a documentation statement that describes a function's behavior should not only be true but also say more than the interface does. Attribution checks this basic requirement. It complements correctness, can guide documentation generation, and can serve as a measure for evaluating documentation and specifications.

\section{Data Availability}
The replication package for this paper will be deposited in a permanent archive. It contains source-code snapshots, experimental inputs and outputs, verification evidence, prompts and model configurations, and scripts for offline recomputation, and it retains the records of selection, exclusion, and unresolved cases. The package also describes the execution environment and the configuration of external services.

\bibliographystyle{ACM-Reference-Format}
\bibliography{Attribution_FSE2027_refs}
\end{document}